\documentclass[a4paper, 11pt]{scrartcl}

\usepackage{amssymb}
\usepackage{amsthm}
\usepackage{tikz}
\usepackage{amsmath}
\usepackage{latexsym}
\usepackage{amsfonts}
\usepackage{color}
\usepackage{url}
\usepackage{algorithm}
\usepackage{algpseudocode}
\newenvironment{keywords}%
   {\begin{trivlist}\item[]{ \textit{Keywords}:}\ }
   {\end{trivlist}}
 
\DeclareMathOperator{\Fact}{Fact} 
\DeclareMathOperator{\tr}{Tr}

\DeclareMathOperator{\sta}{sc}
\DeclareSymbolFont{rsfscript}{OMS}{rsfs}{m}{n}
\DeclareSymbolFontAlphabet{\mathrsfs}{rsfscript}

\DeclareMathOperator{\rc}{rc}
\DeclareMathOperator{\Syn}{Syn}

\DeclareSymbolFont{rsfscript}{OMS}{rsfs}{m}{n}

\newtheorem{theorem}{Theorem}

\newtheorem{defn}{Definition}
\newtheorem{lemma}{Lemma}
\newtheorem{cor}{Corollary}

\newcommand{\rf}{\rightarrow}
\newcommand{\la}{\langle}
\newcommand{\ra}{\rangle}

\def\mapright#1{\smash{\mathop{\longrightarrow}\limits^{#1}}}

\title{Synchronizing automata and the language of minimal reset words}

\author{Emanuele Rodaro\\
Dipartimento di Matematica\\
  Politecnico di Milano\\
  Piazza Leonardo da Vinci, 32\\
  20133 Milano, Italy}

\date{}

\begin{document}
\maketitle

\begin{abstract}
We study a connection between synchronizing automata and its set $M$ of minimal reset words, i.e., such that no proper factor is a reset word. We first show that any synchronizing automaton having the set of minimal reset words whose set of factors does not contain a word of length at most $\frac{1}{4}\min\{|u|: u\in I\}+\frac{1}{16}$ has a reset word of length at most $(n-\frac{1}{2})^{2}$. In the last part of the paper we focus on the existence of synchronizing automata with a given ideal $I$ that serves as the set of reset words. To this end, we introduce the notion of the tail structure of the (not necessarily regular) ideal $I=\Sigma^{*}M\Sigma^{*}$. With this tool, we first show the existence of an infinite strongly connected synchronizing automaton $\mathrsfs{A}$ having $I$ as the set of reset words and such that every other strongly connected synchronizing automaton having $I$ as the set of reset words is an homomorphic image of $\mathrsfs{A}$. Finally, we show that for any non-unary regular ideal $I$ there is a strongly connected synchronizing automaton having $I$ as the set of reset words with at most $(km^{k})2^{km^{k}n}$ states, where $k=|\Sigma|$, $m$ is the length of a shortest word in $M$, and $n$ is the dimension of the smallest automaton recognizing $M$ (state complexity of $M$). This automaton is computable and we show an algorithm to compute it in time $\mathcal{O}((k^{2}m^{k})2^{km^{k}n})$. 
\end{abstract}


\section{Introduction}
In this paper we are interested in automata from their dynamical point of view, and not as languages recognizers. Thus, for us an automaton (for short DFA) is just a tuple $\mathrsfs{A} = \la Q,\Sigma,\delta\ra$, where $Q$ is the set of states, $\Sigma$ is the finite alphabet acting on $Q$, and the function $\delta:Q\times\Sigma\to Q$ describes the action of $\Sigma$ on the set $Q$. In literature these objects are usually called \emph{semiautomata}. We may depict an automaton as a labelled digraph having edges $q\mapright{a}p$ whenever $\delta(q,a)=p$ and with the property of being complete: every $a\in\Sigma$ and $q\in Q$ there is an out-going edge $q\mapright{a}p$, and being deterministic: if $q\mapright{a}p$, $q\mapright{a}p'$ are two edges of $\mathrsfs{A}$, then $p=p'$. Throughout the paper we will use the action notation by putting $q\cdot a=\delta(q,a)$, this action naturally extends to $\Sigma^{*}$ and to the subsets of $Q$ in the obvious way. 
Automata are mostly used in theoretical computer science as languages recognizers: by pinpointing an initial state $q_{0}$ and a set of final states $F\subseteq Q$ the automaton $\mathrsfs{A} $ defines the regular language $L[\mathrsfs{A} ]=\{u\in \Sigma^{*}: q_{0}\cdot u\in F\}$, and every regular language is recognized in this way by a finite automaton, see for instance \cite{hop}. 
The interested in automata from their dynamical point of view is mostly motivated by the longstanding Cerny's conjecture regarding the class of synchronizing automata. These are automata having a word $u\in\Sigma^{*}$, called \emph{reset}, sending all the states to a unique one, i.e., $|Q\cdot u|=1$. Cerny's conjecture states that a $n$-state synchronizing automaton has always a reset word of length at most $(n-1)^{2}$, see \cite{Ce64}. The literature around Cerny's conjecture and synchronizing automata is quite impressive and span from the algorithmic point of view to the proof of Cerny's conjecture or the existence of quadratic bounds on the smallest reset word for several classes of automata, see for instance \cite{AlRo, AnVo, BeSz, BeBePe, Dubuc, Epp, GreKi, Kari, Steinb, Trah, Vo_CIAA07}. The best upper bound for the shortest reset word is cubic $(n^{3}-n)/6$ obtained by Pin-Frankl \cite{Frankl,Pin} and recently improved by Szykula in \cite{Szy17} by a factor of 4/46875. For a general survey on synchronizing automata and Cerny's conjecture see \cite{KaVo, Vo_Survey}.
\\
In this paper we continue the language theoretic approach to synchronizing automata initiated in a series of recent papers \cite{GuMasPribFG, PrincIdFI, SOFSEM,PriRoMinWords, PriRo11,ReRo16,ReRo13}. The starting point of such an approach is a simple observation: the set of reset words is a two-sided ideal (ideal for short) of the free monoid $\Sigma^{*}$ that is also a regular language. The natural questions is whether any given regular ideal $I$ is the set of the reset words of some synchronizing automaton. In \cite{SOFSEM} it is observed that the minimal DFA recognizing $I$ is a synchronizing automaton with a sink state, i.e., having a particular state $s$ with transitions $s\mapright{a}s$, $a\in\Sigma$, whose set of reset words is exactly $I$. This simple observation led the author to introduce a new notion of descriptional complexity for the class of regular ideal languages. The \emph{reset complexity} $\rc(I)$ of an ideal $I$ is the number of states of the smallest synchronizing automaton $\mathrsfs{B}$ for which $I$ serves as the set of reset words of $\mathrsfs{B}$. The interesting fact is that Cerny's conjecture holds if and only if $\rc(I)\ge \sqrt{\| I\|}+1$ with $\| I\|=\min\{|w|: w\in I\}$ holds for any ideal language $I$. This observation justifies the study of Cerny's conjecture and synchronizing automata from a language theoretic point of view instead of the structure of the automaton. In this paper we contribute to this point of view and we underly the importance of the set of minimal words in such an approach. 
Any ideal $I$ has a minimal set of generators $\mathcal{M}(I)$ (called the set of minimal words) such that $I=\Sigma^{*}\mathcal{M}(I)\Sigma^{*}$ and $\mathcal{M}(I)$ is factor-free, i.e., any factor $u$ of any $w\in \mathcal{M}(I)$ does not belong to $\mathcal{M}(I)$. We prove that if we consider a synchronizing automaton with $n$ states having set of factors of length $\ell$ of the set $\mathcal{M}(I)$ ($\Fact_{\ell}(\mathcal{M}(I))$) that does not contain a word $u\in\Sigma^{\ell}$, then there is a reset word of length at most $ \frac{n(n-1)}{2}+2\ell$. From this result, as an immediate corollary, we show that Cerny's conjecture holds if this missing factor has length at most $\frac{1}{4}(n^{2}-3n+2)$. Further, a quadratic bound $(n-\frac{1}{2})^{2}$ holds for all the synchronizing automata with $n$ states having set of reset words $I$ satisfying $\Fact_{\ell}(\mathcal{M}(I))\setminus\Sigma^{\ell}\neq\emptyset$ for some $\ell\le\frac{\|I\|}{4}+\frac{1}{16}$. Roughly speaking, a potential counterexample to the Cerny conjecture should be searched in the class of strongly connected synchronizing automata whose set of reset words is an ideal whose set of minimal words contains as factors all the words of length at most $\frac{1}{4}(n^{2}-3n+2)$.
\\
In the second part of the paper we focus on how to build a strongly connected automaton from a given ideal $I$. We show that also for this problem the set of minimal words of an ideal $I$ plays an important role in building such an automaton. The importance of the class of strongly connected synchronizing automata lies on the fact that proving Cerny's conjecture for this class would led to a full solution of this conjecture. In \cite{ReRo16, ReRo13} it is proved in a non-constructive way that any non-unary regular ideal may serve as the set of reset words of some strongly connected synchronizing automaton. Here we introduce the notion of the tail structure of a (not necessarily) regular ideal $I$. With this tool we first show a very natural way to build a strongly connected (in general infinite) synchronizing automaton having $I$ as the set of reset words. Such an automaton is named the maximal lifted strongly connected automaton since every other strongly connected synchronizing automaton having $I$ as the set of reset words is an homomorphic image of this automaton. Finally, using the tail structure we show a more natural way to build a strongly connected synchronizing automaton having a regular ideal $I$ as the set of the reset words. The advantage of this approach is twofold: we provide an algorithm to effectively build such an automaton and the upper bound on the number of states of this automaton is not a double exponential like the bound presented in \cite{ReRo16,ReRo13} but it is $(km^{k})2^{km^{k}n}$ where $k=|\Sigma|$, $m=\|I\|$, and $n$ is the number of states of the minimal DFA accepting the language $\mathcal{M}(I)$ (the state complexity of $\mathcal{M}(I)$, see for instance \cite{BrzJiLi}). The importance of finding ``small'' strongly connected synchronizing automata having a certain regular ideal $I$ as the set of reset words, is justified by the fact that, if Cerny's conjecture holds, then we cannot find such an automaton with a number of states less than $\sqrt{\| I\|}+1$ states. Understanding how to build these automata and why they cannot be less than $\sqrt{\| I\|}+1$ states is a fundamental question.

\section{Some notations}\label{sec: prel}
We collect here some of the notations and basic definitions that will be used throughout the paper. We assume that the finite alphabet $\Sigma$ has more than one element. For $u,v\in \Sigma^{*}$ we say that $u$ is a prefix, suffix, factor of $v$ whenever $v=uu'$, $v=u'u$, $v=u'uu''$ for some $u', u''\in\Sigma^{*}$, respectively. Let $u\in\Sigma^{*}$, for an index $0\le i\le |u|$ we denote by $u[i:]$ ($u[:i]$) the suffix (prefix) of $u$ of length $i$; note that $u[0:]=u[:0]$ is the empty word $\varepsilon$. The $i$-th element of $u$ is denoted by $u[i]$.
For a language $L\subseteq \Sigma^*$ and $u\in\Sigma^*$, we put $Lu=\{xu:x\in L\}$, $uL=\{ux:x\in L\}$. The left (right) quotient of $L$ with respect to $u$ is the set $u^{-1}L=\{v\in\Sigma^{*}:uv\in L\}$ ($Lu^{-1}=\{v\in\Sigma^{*}:vu\in L\}$).
\\
An ideal $I$ on an alphabet $\Sigma$ is a language satisfying $\Sigma^{*}I\Sigma^{*}\subseteq I$. A left (right) ideal is a set $I$ satisfying $\Sigma^{*}I\subseteq I$ ($I\Sigma^{*}\subseteq I$). For an ideal $I$ we denote by $\mathcal{M}(I)$ the set of minimal synchronizing words, i.e. $\mathcal{M}(I)=I\setminus (\Sigma^{+}I\cup I\Sigma^{+})$. Note that $\mathcal{M}(I)$ is the set of generators of the ideal $I$ in the free semigroup $\Sigma^{*}$, i.e., $I=\Sigma^{*}\mathcal{M}(I)\Sigma^{*}$. Throughout the paper we put $\|I\|=\min\{|u|:u\in I\}$, 
note that $\|I\|=\min\{|u|: u\in \mathcal{M}(I)\}$.
Let $\mathrsfs{A}=\la Q,\Sigma,\delta\ra$ be a synchronizing automaton, let $\Syn(\mathrsfs{A})$ denote the set of reset words. As we have already noted $I=\Syn(\mathrsfs{A})$ is an ideal which is also regular since it is recognized by the power automaton of $\mathrsfs{A}$ with initial state $Q$ and final set of states all the singletons $\{q\}$, $q\in Q$. Since in this paper we are dealing with ideals that arise from the set of reset words of a synchronizing automaton, most of the time we will assume $I$ regular, except in Section \ref{sec: reset right decomposition} when we generalize the notion of reset left regular decomposition and we introduce the maximal lifted strongly connected automaton associate to a (non necessarily regular) ideal.
\\
An \emph{automaton homomorphism} (simply an homomorphism) $\varphi:\mathrsfs{A}\rf\mathrsfs{B}$ between the two semiautomata $\mathrsfs{A}=\la Q,\Sigma,\delta\ra$, $\mathrsfs{B}=\la T,\Sigma,\xi\ra$ is a map $\varphi:Q\rf T$ preserving the actions, i.e. $\varphi(\delta(q,a))=\xi(\varphi(q),a)$ for all $a\in\Sigma$.

\section{Cerny's conjecture for ideals having minimal reset words with a short missing factor}
In this section we prove quadratic bounds and Cerny's conjecture for synchronizing automata with set of minimal reset words that does not have a sufficiently small factor. For a language $L\subseteq \Sigma^{*}$, the set of factors of length $\ell$ is denoted by $\Fact_{\ell}(L)$. In case $\Fact_{\ell}(L)\neq \Sigma^{\ell}$ any word $ \Sigma^{\ell}\setminus \Fact_{\ell}(L)$ is called a \emph{missing factor}.
\\
For a set $H\subseteq \Sigma^{*}$ and a word $u\in \Sigma^{*}$ we denote by $u\wedge_{s} H$ the maximal suffix $h$ of $u$ such that $h$ is a prefix of some word in $H$. 
\begin{theorem}\label{theo: missing factor}
Let $\mathrsfs{A}=\langle Q,\Sigma,\delta\rangle$ be a synchronizing automaton with set of minimal reset words $M=\mathcal{M}(\Syn(\mathrsfs{A}))$. If $\Fact_{\ell}(M)\neq \Sigma^{\ell}$, then there is a reset word $u$ with
$$
|u|\le \frac{n(n-1)}{2}+2\ell
$$
\end{theorem}
\begin{proof}
Let $I=\Syn(\mathrsfs{A})$ and consider the following set 
$$
\Lambda_{\ell}=\left\{H=Q\cdot v\mbox{ for some }  v\in\Sigma^{*}\setminus I\mbox{ such that }|v\wedge_{s} M|\le  \frac{n(n-1)}{2}+\ell\right\}
$$
Let $Z\in\Lambda_{\ell}$ be an element with minimal cardinality, and let $u\in\Sigma^{*}$ such that $Z=Q\cdot u$ with $ |u\wedge_{s} M|\le  \frac{n(n-1)}{2}+\ell$. Let $x\in \Sigma^{\ell}\setminus \Fact_{\ell}(M)$. If $|Z\cdot x|=1$, then $u\in I$. Indeed, we clearly have that $ux\in I$ since $|Q\cdot (ux)|=|Z\cdot x|=1$. Thus, there is a minimal reset word $m\in M$ such that $m$ is a factor of $ux$. Since $u\notin I$, $m$ is not a proper factor of $u$, and so there is a suffix $v$ of $u$ which is a prefix of $m$ so that $m$ is a prefix of $vx$. By definition, since $Z\in\Lambda_{\ell}$ and $Z=Q\cdot u$ we have:
$$
|v|\le |u\wedge_{s} M|\le  \frac{n(n-1)}{2}+\ell
$$
Hence, we have a reset word $m$ of length $|m|\le |vx|\le  n(n-1)/2+\ell+|x|\le n(n-1)/2+2\ell$ and we are done. Therefore, we may assume that $|Z\cdot x|\neq 1$. Since $|Z\cdot x|>1$, we may collapse a pair of states of $Z\cdot x$. Hence, we may find a word $w$ with 
$$
|w|\le \frac{n(n-1)}{2}
$$
such that $|Z\cdot xw|<|Z\cdot x|$. We consider the following two cases:
\begin{itemize}
\item $|Z\cdot xw|=1$, whence $uxw\in I$, and let $m\in M$ be a minimal reset word which is a factor of $uxw$. Note that $m$ is not a factor of $ux$, for if $ux\in I$ would imply $|Z\cdot x|=1$, a contradiction. Since $x$ is not a factor of any $m\in M$, then necessarily $m$ is a factor of $xw$. Hence, there is a reset word $m$ of length:
$$
|m|\le |x|+|w|\le \ell+\frac{n(n-1)}{2}
$$
and we are done. 
\item $|Z\cdot xw|>1$. We claim that $Z\cdot xw\in \Lambda_{\ell}$. Indeed, first note that $uxw$ is not reset and thus $uxw\in\Sigma^{*}\setminus I$. We claim that if $uxw\wedge_{s} h=t$, for some $h\in M$, then necessarily $xw\wedge_{s} h=t$. Indeed, if there is a word $u'\in\Sigma^{*}$ such that $u=u''u'$ and $u'xw$ is a prefix of $h$, then $x\in \Fact_{\ell}(M)$, a contradiction. Therefore, we have:
$$
uxw\wedge_{s}M=xw\wedge_{s}M
$$
and so:
$$
|uxw\wedge_{s}M|= |xw\wedge_{s}M|\le  |x|+|w|\le \ell+\frac{n(n-1)}{2}
$$
from which we may conclude that $Z'=Z\cdot xw=Q\cdot uxw\in \Lambda_{\ell}$. However, the inequality $|Z'|<|Z\cdot x|\le |Z|$ contradicts the minimality of the cardinality of $Z\in \Lambda_{\ell}$.
\end{itemize}

%
%
%
%
%
%
%
%
%
\end{proof}
From the previous theorem we derive the following results pointing toward a positive solution of Cerny's conjecture.
\begin{cor}
The following facts hold:
\begin{itemize}
\item If a synchronizing automaton $\mathrsfs{A}$ with $n$ states has a set of minimal synchronizing words with a missing factor of length at most $\frac{1}{4}(n^{2}-3n+2)$, then Cerny's conjecture holds;
\item For any regular ideal $I$ having a set of minimal elements $\mathcal{M}(I)$ not containing a factor of length $\ell\le \frac{\|I\|}{4}+\frac{1}{16}$, then there is a reset word of length at most $(n-\frac{1}{2})^{2}$.
\end{itemize}
\end{cor}
\begin{proof}
The first statement follows directly from Theorem \ref{theo: missing factor}. Regarding the second one, we clearly have 
$$
\|I\|\le \frac{n(n-1)}{2}+\frac{\|I\|}{2}+\frac{1}{8}
$$
hence: $\|I\|\le n(n-1)+\frac{1}{4}$, from which we derive $n\ge \frac{1}{2}+\sqrt{\|I\|}$, whence the minimal reset word is at most $(n-\frac{1}{2})^{2}$.
\end{proof}

\section{The reset left (right) decomposition of an ideal}\label{sec: reset right decomposition}

The notion of reset left (right) regular decomposition of a regular ideal has been introduced in \cite{ReRo13, ReRo16} as a intermediate step to show that such an ideal on a non-unary alphabet may serve as the set of reset words of some strongly connected synchronizing automaton. 
\begin{defn}\label{defn: regular dec}
  A reset left regular decomposition of a regular ideal $I$ is a finite collection $\mathcal{I}=\{I_1, \ldots I_{k}\}$ of disjoint left ideals $I_i$ of $\Sigma^*$ that partitions $I$ and satisfying  
  \begin{itemize}
  \item[i)] For any $a\in\Sigma$ and $I_{i}\in \mathcal{I}$, there is a $I_{j}\in \mathcal{I}$
    such that $I_ia\subseteq I_j$.
  \item[ii)]  For any $u\in\Sigma^*$ if $Iu\subseteq I_i $, for some $I_{i}\in \mathcal{I}$, then $u\in I$.
  \end{itemize}
\end{defn}
There is a categorical equivalence between this notion and the class of strongly connected synchronizing automata. Denote by $\textbf{RLD}_{\Sigma}$ the category of the reset left regular
decompositions on the alphabet $\Sigma$, where an arrow between two objects $f:\mathcal{I}\to \mathcal{J}$ is any map $f:\mathcal{I}\to \mathcal{J}$ such that $I_{i}\subseteq f(I_{i})$ for all $I_{i}\in\mathcal{I}$. We also  consider the category $\textbf{SCSA}_{\Sigma}$ formed by the class of synchronizing automata on the same alphabet $\Sigma$, where $\varphi:\mathrsfs{A}\rightarrow\mathrsfs{B}$ is an arrow if $\varphi$ is an (automaton) homomorphism. Note that any homomorphism between strongly connected automata is necessarily surjective. The importance of these two categories lies in Theorem \cite[Theorem 2.2]{ReRo16} where it is proved that an ideal (regular) language $I$ serves as the set of reset words of some strongly connected synchronizing automaton if and only if it has a reset left regular decomposition. Moreover $\textbf{RLD}_{\Sigma}$ and $ \textbf{SCSA}_{\Sigma}$ are equivalent categories via the two functors $\mathcal{A},\mathcal{D}$ defined by:
  \begin{itemize}
  \item   
   $\mathcal{A}:\textbf{RLD}_{\Sigma}\rf \textbf{SCSA}_{\Sigma}$ defined by
  $$
  \mathcal{A}:\mathcal{I}=\{I_1,\ldots, I_{k}\}\mapsto \mathcal{A}(\mathcal{I})=\la \mathcal{I},\Sigma,\eta\ra 
  $$
  with $\eta(I_i,a)=I_j$ for $a\in\Sigma$ if and only if
  $I_{i}a\subseteq I_{j}$, and if $f:\mathcal{I}\rightarrow \mathcal{J}$ then $\mathcal{A}(f)$ is the homomorphism $\varphi:\mathcal{A}(\mathcal{I})\rightarrow \mathcal{A}(\mathcal{J})$ defined by $\varphi(I_{i})=f(I_{i})$.
  \item $\mathcal{D}:\textbf{SCSA}_{\Sigma}\rf\textbf{RLD}_{\Sigma}$ defined by
  $$
  \mathcal{D}:\mathrsfs{A}=\la Q,\Sigma,\delta\ra\mapsto
   \mathcal{I}(\mathrsfs{A})=\left\{I_{q}=\{u\in\Sigma^{*}:\delta(Q,u)=q\}, q: Q\right\}
  $$
 and if $\varphi:\mathrsfs{A}\rf\mathrsfs{B}$ is an arrow between $\mathrsfs{A}=\la Q,\Sigma,\delta\ra$ and $\mathrsfs{B}=\la T,\Sigma,\xi\ra$, then $\mathcal{D}(\varphi)$ is the arrow defined by $f:\mathcal{I}(\mathrsfs{A})\rf \mathcal{I}(\mathrsfs{B})$ which sends $I_{q}\mapsto I_{\varphi(q)}$. 
\end{itemize}
In the same paper it is proved that any regular ideal on a non-unary alphabet admits a reset left regular decomposition, from which we immediately get that any ideal on a non-unary alphabet serves as the set of reset words of some strongly connected synchronizing automaton. This fact is a key ingredient that transfer the study of strongly connected synchronizing automata to the study of reset left regular decomposition of an ideal. Let us now extend this notion to a non necessarily regular ideal. The finiteness condition in the definition of reset left regular decomposition may be removed, but we need to extend the notion of automaton and consider automata with an infinite number of states. In this setting we still call an automaton a tuple $\mathrsfs{A}=\la Q,\Sigma,\delta\ra$, where $\Sigma$ is still finite but the set of states $Q$ may be infinite (countable many), and $\delta:Q\times \Sigma\to \Sigma$ is the usual transition function. We may still use the action-notation and when $\delta$ is clear from the context we will denote $q\cdot a$, $q\in Q$, $a\in\Sigma$, the state $\delta(q,a)$. This action is extended to the powerset of $Q$ in the obvious way. A semiautomaton is still called strongly connected whenever for any $q,p\in Q$ there is a finite word $w\in \Sigma^{*}$ such that $q\cdot w=p$. Moreover, the notion of synchronizing automaton is unaffected considering infinite semiautomata and we still call $\mathrsfs{A}$ synchronizing if there exists a word $w\in \Sigma^{*}$ such that $|Q\cdot w|=1$. The set of reset words $\Syn(\mathrsfs{A})$ is still an ideal, but it is not, in general, regular. In this setting we may define a \emph{reset left decomposition} as a reset left regular decomposition without the finiteness condition. 
\begin{defn}\label{defn: regular dec}
  A reset left decomposition of a ideal $I$ is a countable collection $\mathcal{I}$ of disjoint left ideals of $\Sigma^*$ that partitions $I$ and it satisfies the following conditions:  
  \begin{itemize}
  \item[i)] For any $a\in\Sigma$ and $J\in \mathcal{I}$, there is a $J'\in \mathcal{I}$
    such that $Ja\subseteq J'$.
  \item[ii)]  For any $u\in\Sigma^*$ if $Iu\subseteq J$, for some $J\in \mathcal{I}$, then $u\in I$.
  \end{itemize}
\end{defn}
With a slight abuse of notation, we still denote by $\textbf{RLD}_{\Sigma}$ and $ \textbf{SCSA}_{\Sigma}$ the categories of the reset left decompositions and of strongly connected synchronizing automata (with a possible infinite number of states). The morphisms in these categories are similar to the morphisms defined in the finite case mentioned before. The finiteness condition in the reset left regular decomposition corresponds to the finiteness of the associated strongly connected synchronizing automaton, and since the finiteness condition is not used in the proof of \cite[Theorem 2.2]{ReRo16}, then using verbatim this proof it is possible to prove the following analogous theorem.
\begin{theorem}\label{theo: general characterization}
An ideal $I$ is the set of reset words of some strongly connected synchronizing semiautomaton if and only if it admits a reset left decomposition. Moreover the categories $\textbf{RLD}_{\Sigma}$ and $ \textbf{SCSA}_{\Sigma}$ are equivalent via the two functors $\mathcal{A},\mathcal{D}$ defined by:
  \begin{itemize}
  \item   
   $\mathcal{A}:\textbf{RLD}_{\Sigma}\rf \textbf{SCSA}_{\Sigma}$ which sends 
  $$
  \mathcal{A}:\mathcal{I}\mapsto \mathcal{A}(\mathcal{I})=\la \mathcal{I},\Sigma,\eta\ra 
  $$
  with $\eta(J,a)=J'$ for $a\in\Sigma$ if and only if
  $Ja\subseteq J'$, and if $f:\mathcal{I}\rightarrow \mathcal{J}$ then $\mathcal{A}(f)$ is the homomorphism $\varphi:\mathcal{A}(\mathcal{I})\rightarrow \mathcal{A}(\mathcal{J})$ defined by $\varphi(J)=f(J)$.
  \item $\mathcal{D}:\textbf{SCSA}_{\Sigma}\rf\textbf{RLD}_{\Sigma}$ defined by
  $$
  \mathcal{D}:\mathrsfs{A}=\la Q,\Sigma,\delta\ra\mapsto
   \mathcal{I}(\mathrsfs{A})=\left\{J_{q}=\{u\in\Sigma^{*}:\delta(Q,u)=q\}, q: Q\right\}
  $$
 and if $\varphi:\mathrsfs{A}\rf\mathrsfs{B}$ is an arrow between $\mathrsfs{A}=\la Q,\Sigma,\delta\ra$ and $\mathrsfs{B}=\la T,\Sigma,\xi\ra$, then $\mathcal{D}(\varphi)$ is the arrow defined by sending $J_{q}\to J_{\varphi(q)}$.
\end{itemize}
\end{theorem}
Proving the existence of a reset left \textbf{regular} decomposition of a regular language is in general an hard task. However, in the next section we show a natural and easy way to decompose a general ideal into a reset left decomposition using the notion of tail structure of an ideal. As an unexpected consequence, we show that every ideal serves as the set of a reset words of a strongly connected (countably infinite) synchronizing automaton.  
The notion of tail structure also allows us to give a more natural and explicit (with respect to the one presented in \cite{ReRo13, ReRo16}) way to decompose a regular ideal by a reset left \textbf{regular} decomposition. 

\section{Tail structure and reset left (regular) decomposition.}
Throughout this section $M=\mathcal{M}(I)$ will denote the language of the minimal reset words of an ideal $I$. 
\begin{defn}[tail structure]
For an element $u\in I$, the \emph{last factor} $\lambda(u)\in M$ is the minimal reset word such that
$$
u=u'\lambda(u)v
$$
and if $\lambda(u)=az$ for some letter $a\in\Sigma$, $z\in\Sigma^{*}$, then $zv$ does not contain any word of $I$ (and so of $M$) as a factor; the word $zv$ in the above factorization is called the \emph{tail} of $u$ and it is denoted by $\tau(u)$. Equivalently, the tail $\tau(u)$ is the maximal suffix of $u$ that does not belong to $I$, see the figure below. The ordered pair $\sigma(u)=(\lambda(u), \tau(u))$ is called \emph{the tail structure} of $u$.
\end{defn}
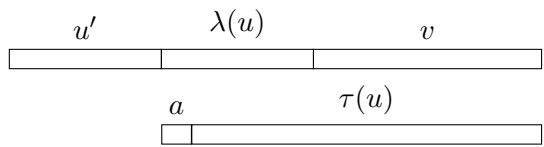
\begin{center}
\begin{tikzpicture}
\begin{scope}[xshift=0]
\node [draw,rectangle,minimum width=2cm,minimum height=0.2cm, label=$u'$] {};
\end{scope}
\begin{scope}[xshift=2cm]
\node [draw,rectangle,minimum width=2cm,minimum height=0.2cm, label=$\lambda(u)$] {};
\end{scope}

\begin{scope}[xshift=4.5cm]
\node [draw,rectangle,minimum width=3cm,minimum height=0.2cm, label=$v$] {};
\end{scope}

\begin{scope}[yshift=-1cm, xshift=1.2cm]
\node [draw,rectangle,minimum width=0.39cm,minimum height=0.2cm, label=$a$] {};
\end{scope}

\begin{scope}[yshift=-1cm, xshift=3.7cm]
\node [draw,rectangle,minimum width=4.6cm,minimum height=0.2cm, label=$\tau(u)$] {};
\end{scope}

\end{tikzpicture}
\end{center}
Note that the factorization $u=u'\lambda(u)v$ is univocally determined: if $h_{1}, h_{2}\in M$ are two words satisfying the condition of the definition of last factor, then either $h_{1}$ is a prefix of $h_{2}$, or vice-versa, however, this contradicts the fact that $h_{1}, h_{2}\in M$. For technical reasons that will be clear later, we extend the notion of tail to all the words $u\in\Sigma^{*}\setminus I$ by putting $\tau(u)=u$ and for the tail structure we put $\sigma(u)=(\varepsilon,u)$.

\subsection{The maximal lifted strongly connected automaton}
In this section we show that every strongly connected synchronizing automaton $\mathrsfs{A}$ with $\Syn(\mathrsfs{A})=I$ is an homomorphic image of a unique (infinite) strongly connected synchronizing automaton $\mathrsfs{L}(I)$, called \emph{the maximal lifted strongly connected automaton}.
The tail structure is a natural starting point to exhibit a reset left decomposition of an ideal. Indeed, consider the set of all the possible tail structures of the elements of $I$: $T=\{(x,y): \sigma(u)=(x,y)\mbox{ for some } u\in I\}$. Note that, in case $I$ is regular, the union of all the tails of the elements of $I$ is a regular language:
$$
\bigcup_{u\in I}\tau(u)=\bigcup_{a\in\Sigma}(a^{-1}M\Sigma^{*}\setminus I)
$$
For any $(x,y)\in T$ we define  
$$
I(x,y)=\{v\in I: \sigma(v)=(x,y)\}
$$
It is not difficult to see that if the set $I(x,y)$ is non-empty, and $x=az$, for some $a\in\Sigma, z\in\Sigma^{*}$, then 
$$
I(x,y)=\Sigma^{*}ay
$$
Moreover, these left ideals form a partition of the ideal $I$, thus we may consider the \emph{tail structure decomposition} of $I$:
$$
\mathrsfs{T}(I)=\left\{I(x,y): (x,y)\in T\right\}
$$
This family forms a reset left decomposition as the following theorem shows.
\begin{theorem}\label{theo: maximal lifted}
With the notation above, the family $\mathrsfs{T}(I)$ is a reset left decomposition of $I$. Moreover, for any other  reset left decomposition $\mathrsfs{D}$ of $I$ there is an epimorphism $\varphi:\mathrsfs{T}(I)\to \mathrsfs{D}$ defined by 
$$
\varphi(I(x,y))=J\mbox{ whenever there is }J\in\mathcal{D}\mbox{ such that }I(x,y)\subseteq J
$$
\end{theorem}
\begin{proof}
We have already remarked that $\mathrsfs{T}(I)$ is a family of left ideals decomposing $I$. Moreover, for any $a\in\Sigma$, and $(bx,y)\in T$, for some $b\in\Sigma$, $I(ax,y)a\subseteq I(x',y')$ where $(x',y')=\sigma(bya)$. Before proving the reset condition ii) of Definition \ref{defn: regular dec}, we show the existence of the epimorphism $\varphi:\mathrsfs{T}(I)\to \mathrsfs{D}$. Take an arbitrary $I(az,y)$, for some $a\in\Sigma$ $z\in\Sigma^{*}$, since $I(az,y)=\Sigma^{*}ay$ and $ay\in M\Sigma^{*}$, then $ay$ belongs to some $J\in\mathcal{D}$. Hence, being $J$ a left ideal, we get that $I(az,y)=\Sigma^{*}ay\subseteq J$. This shows that the map $\varphi$ defined in the statement of the lemma is well defined. The surjectivity follows from the fact that for any $J$ and any $u\in J$, if $\sigma(u)=(x,y)$, then $\varphi(I(x,y))=J$. The fact that $\varphi$ is an homomorphism follows from $\varphi(I(x,y)a)\subseteq Ja\subseteq J'$, for some $J'\in\mathcal{D}$, whence $\varphi(I(x,y)a)=\varphi(I(x,y))a$. We are now in position to prove the reset condition. Indeed, if $Iu\subseteq I(x,y)$, for some $(x,y)\in T$, then we have 
$$
J'=\varphi(I(x,y))=\varphi(I(z,t)u)=\varphi(I(z,t))u
$$ 
for some $J'\in \mathcal{D}$ and for all $(z,t)\in T$. This last condition implies $J u\subseteq J'$ for all $J\in\mathcal{D}$, hence $u\in I$ since $\mathrsfs{D}$ is a reset left decomposition. 
\end{proof}
From Theorem \ref{theo: general characterization} we may consider the semiautomaton $\mathrsfs{L}(I)=\mathrsfs{A}(\mathrsfs{T}(I))$ associated to the reset left decomposition $\mathrsfs{T}(I)$, called the \emph{maximal lifted strongly connected automaton}. This name is justified by the following corollary.
\begin{cor}
Let $I$ be an ideal, and let $\mathrsfs{L}(I)$ be the maximal lifted strongly connected automaton. Then, any other strongly connected synchronizing automaton $ \mathrsfs{A}$, with $\Syn(\mathrsfs{A})=I$, is an homomorphic image of $\mathrsfs{L}(I)$, i.e., there is an epimorphism $\varphi: \mathrsfs{L}(I)\to \mathrsfs{A}$. Moreover, this automaton is unique (up to automaton-isomorphisms) in the following sense: if $\psi:\mathrsfs{A}\to \mathrsfs{L}(I)$ is an epimorphism, then $\psi$ is an isomorphism with inverse $\varphi$. 
\end{cor}
\begin{proof}
The first statement is an immediate consequence of Theorem \ref{theo: maximal lifted} and Theorem \ref{theo: general characterization}.
By Theorem \ref{theo: general characterization} we may identify $\mathrsfs{A}$ with its reset left decomposition $\mathcal{I}$. Thus any state $q$ of $\mathrsfs{A}=\la Q,\Sigma, \delta\ra$ may be identified with the left ideal 
$$
I_{q}=\{u\in\Sigma^{*}: Q\cdot u=q\}
$$
We claim that $\varphi(\psi(q))=q$. By Theorem \ref{theo: general characterization} state $\psi(q)$ of $\mathrsfs{L}(I)$ may be identified with a left ideal $I(x,y)$ for suitable words $x,y \in\Sigma^{*}$. Now, by definition of the map $\varphi$, $\varphi(\psi(q))$ is a state $p$ corresponding to a left ideal $I_{p}$ of $\mathcal{I}$ with the property $I(x,y)\subseteq I_{p}$. Since $\psi$ is an epimorphism, then all the reset words that sends all the states of $\mathrsfs{A}$ in $q$ are also sending all the states of $\mathcal{L}(I)$ to $\psi(q)$, whence by Theorem \ref{theo: general characterization} we get $I_{q}\subseteq I(x,y)$. Hence, $I_{q}\subseteq I(x,y)\subseteq I_{p}$ that implies $I_{p}=I_{q}$ by definition of reset left decomposition. Thus, by Theorem \ref{theo: general characterization} we get $q=p$. Therefore, $\psi$ is injective and so it is an isomorphism with inverse $\varphi$.
\end{proof}
This corollary answers the question posed in Problem 4 \cite{ReRo16} regarding the existence of arbitrarily large strongly connected synchronizing automaton having an ideal $I$ as the set of its reset words. Even if we consider a principal ideal $P=\Sigma^{*}aw\Sigma^{*}$, for some $aw\in\Sigma^{*}$, the reset left decomposition $\mathrsfs{T}(P)$ is formed by left ideals $I(aw,y)$ where $y$ run on the set $(w\Sigma^{*}\setminus P)$, that is clearly infinite.

\subsection{A construction of a strongly connected synchronizing automaton with a given set of reset words}
Using the tail structure introduced in the previous section, we now show a more natural way to decompose a regular ideal into a reset left \textbf{regular} decomposition then the one devised in \cite{ReRo13, ReRo16}. 
Since $M=\mathcal{M}(I)$ is regular, we may consider the minimal DFA $\mathrsfs{B}=\la X,\Sigma, \delta, q_{0}, F\ra$ recognizing $M$. Note that, in general, $\mathrsfs{B}$ is not complete, and in case a state $q$ is not defined on a letter $a\in\Sigma$ we say that $q\cdot a$ is undefined. Note that $F$ is formed by a unique state. Indeed, using the definition of the set of minimal reset words $M$, it is not difficult to check that for any state of $f\in F$, $f\cdot a$ is undefined for every $a\in\Sigma$. Thus, by minimality we deduce that $|F|=1$ and we denote this unique final state by $f$.
The \emph{set of visiting states} associated to a word $u\in\Sigma^{*}$ with $|u|=n$ is the subset of $X$ defined by 
$$
\nu(u)=\{q_{0}\cdot u[i:], i=0,\ldots, n\}
$$
Note that this set corresponds to the states of the DFA obtained by the powerset construction from the NFA formed by $\mathrsfs{B}$ by adding the transitions $q_{0}\mapright{a} q_{0}$, $a\in\Sigma$. This NFA recognizes the language $\Sigma^{*}M$. The set of visiting states has the following property.
\begin{lemma}\label{lem: unique final state}
With the above notation, for any $u\in \Sigma^{*}$ let $\sigma(u)=(ax,y)$, for some $a\in(\Sigma\cup\{\varepsilon\})$, $x,y\in\Sigma^{*}$. If $f\in \nu(u)$, then $f=q_{0}\cdot \lambda(u) $, and this occurs if and only if $y=x$, i.e., $u\in \Sigma^{*}M$. In general, we have:
$$
\nu(u)\setminus F=\nu(\tau(u))
$$
\end{lemma}
\begin{proof}
Let $\sigma(u)=(ax,y)$ so that $u=way$, for some $w\in\Sigma^{*}$. Note that $q_{0}\cdot u[i:]$ is undefined for all the suffixes $u[i:]=w'ay$, for some $w'\in\Sigma^{+}$. Indeed, this follows from $w'ay\Sigma^{*}\cap M=\emptyset$, and by definition of the set of minimal reset words $M$ and the definition of the minimal DFA $\mathrsfs{B}$ recognizing $M$. Therefore, if $k$ is the integer such that $u[k:]=ay$, then the states forming $\nu(u)$ are 
$$
q_{0}\cdot u[j:], \,j\ge k
$$
Hence $f\in \nu(u)$ if and only if $u[k:]=ax\in M$, and in this case 
$$
f=q_{0}\cdot ax=q_{0}\cdot \lambda(u)
$$
and this occurs if and only if $y=x$.
In case $u[k:]\in ax\Sigma^{+}$ we have that $q_{0}\cdot u[k:]$ is undefined, and $q_{0}\cdot u[j:]\neq f$ for all $j\ge k$. Hence, in both the two cases $u[k:]= ax$ and $u[k:]\in ax\Sigma^{+}$, we have:
$$
\nu(u)\setminus F= \{q_{0}\cdot u[j:],\,j<k\}= \nu(\tau(u))
$$
\end{proof}
Note that, by the above remark, there is a natural action of $\Sigma$ on the set of visiting states. If $\nu(u)=\{p_{1}, \ldots, p_{k}\}$ then for any $a\in\Sigma$ we define:
$$
\nu(u)\circ a=\nu(ua)=\{p_{1}\cdot a, \ldots, p_{k}\cdot a,q_{0}\}
$$
Now, we need to add extra information to $\nu(u)$ to keep trace of the tail structure and the way it evolves. This is crucial at a certain point to prove the reset condition ii) of Definition \ref{defn: regular dec}. Let $m=\|I\|$, and let $\mathbb{Z}_{m}$ be the ring of the integers module $m$. Let $Z_{m}[\Sigma]$ be the free module over the ring $\mathbb{Z}_{m}$ generated by $\Sigma$. For a word $u\in\Sigma$ the \emph{trace} of $u$ is the element of $Z_{m}[\Sigma]$ defined by:
$$
\tr(u)=\sum_{a\in\Sigma}p_{a}a
$$
where $p_{a}\in\mathbb{Z}_{m}$ is the number of occurrences of the letter $a$ in $u$ module $m$. For example, $\tr(a^{3}bac^{4}a)=b+a$ for $m=4$ since $p_{c}=0\mod 4$ and $p_{a}=p_{b}=1\mod 4$.

Consider the alphabet $A=\Sigma\times \mathbb{Z}_{m}[\Sigma]\times X$, where we recall that $X$ is the set of states of $\mathrsfs{B}$. For a generic word $u\in \Sigma^{*}$ with $\sigma(u)=(ax,y)$, for some $a\in(\Sigma\cup\{\varepsilon\})$ and $x,y\in\Sigma^{*}$, if $v=ay$ we consider the subset $\omega(u)\in 2^{A}$ defined by
$$
\omega(u)=\left\{\left(v[i], \tr(v[i:]), q_{0}\cdot v[i:]\right), i=0,\ldots, |v|\right\}
$$
with the convention that if $q_{0}\cdot v[i:]$ is not defined, then the corresponding element $(y[i], \tr(y[i:]), q_{0}\cdot y[i:])$ is set to be empty. Note that the projection of $\omega(u)$ into the last component of the alphabet $A$ is exactly $\nu(u)$. Roughly speaking $\omega(u)$ encodes the information of the tail structure of $u$ by keeping trace of the first letter of a suffix of the tail $\tau(u)$, the number of occurrences of the letters occurring in that suffix and the state reached by applying this suffix. Note that $\omega(u)$ contains at least the element $(\varepsilon, 0, q_{0})$. If $f=\nu(u)\cap F\neq \emptyset$, which corresponds by Lemma \ref{lem: unique final state} to the case $u\in \Sigma^{*}M$, we have in $\omega(u)$ a unique element containing the final state of the form $(a,x, f)$, that we denote by $\varphi(\omega(u))$. Note that by the same Lemma \ref{lem: unique final state} it is not difficult to see that
\begin{equation}\label{eq: final state}
\omega(u)\setminus (a,x,f)=\omega(\tau(u))
\end{equation}
that corresponds exactly to the case $u\in \Sigma^{*}M$, while in case $u\notin\Sigma^{*}M$ we get:
$$
\omega(u)=\omega(\tau(u))
$$
We consider the set $\Omega(I)=\{\omega(u): u\in \Sigma^{*}\}$. Note that if $m=\|I\|$, $k=|\Sigma|$, $n=|X|$, a rough upper bound on the cardinality of $\Omega$ is:
$$
|\Omega(I)|\le 2^{km^{k}n}
$$
The following lemma is a consequence of the definitions.
\begin{lemma}\label{lem: equality of omega}
There is an action of $\Sigma$ on $\Omega(I)$ defined by: for all $\omega(u)\in \Omega(I)$ and $a\in\Sigma$:
$$
\omega(u)\circ a=\omega(ua)
$$
More explicitly, we have
$$
\omega(u)\circ a=\left\{(c,s+a, p\cdot a): (c,s,p)\in \omega(u)\setminus \{(\varepsilon, 0, q_{0})\}\right\} \cup\{(a,a, q_{0}\cdot a), (\varepsilon, 0, q_{0})\}
$$
In particular, if $\omega(u)=\omega(v)$, then $\omega(u)\circ a=\omega(v)\circ a$ for all $a\in\Sigma$.
\end{lemma}
We now define the \emph{tail action}.
\begin{defn}[Tail action]
Consider the action of the alphabet $\Sigma$ on the set 
$$
T(I)=\Sigma\times \mathbb{Z}_{m}[\Sigma]\times\Omega(I)
$$
defined in the following way: for all $a\in\Sigma$ and $(b,x,\omega(u))\in T(I)$ we put:
\begin{align*}
&(b,x,\omega(u))\cdot a=\\
&=\begin{cases}
(b,x+a,\omega(u)\circ a), \mbox{ if }ua\notin \Sigma^{*}M\\
\left(c,s,\omega(\tau(ua))\right), \mbox{ if }ua\in \Sigma^{*}M\mbox{ and }\varphi(\omega(ua))=(c,s,f)
\end{cases}
\end{align*}
\end{defn}
For any $(b,x,\omega(u))\in T(I)$ we define the following set:
$$
I(b,x,\omega(u))=\{v\in I: \sigma(v)=(bw,t), x=\tr(bt), \omega(u)=\omega(v)\}
$$
We have the following lemma.
\begin{lemma}\label{lem: action of left ideals}
For any $a\in\Sigma$ we have:
$$
I(b,x,\omega(u))a\subseteq I\left((b,x,\omega(u))\cdot a\right)
$$
\end{lemma}
\begin{proof}
Let $v\in I(b,x,\omega(u))$ with $\sigma(v)=(bw,t)$. We have to show that $va\in I\left((b,x,\omega(u))\cdot a\right)$. We consider the following two cases.
\begin{itemize}
\item $ua\in\Sigma^{*}M$. Since $\omega(u)=\omega(v)$ by Lemma \ref{lem: equality of omega} we get $\omega(ua)=\omega(va)$. Moreover, since $ua\in\Sigma^{*}M$, by Lemma \ref{lem: unique final state}, we get $f\in\nu(ua)$, and there is a unique element $\varphi(\omega(ua))=(c,s,f)\in \omega(ua)=\omega(va)$. Hence, $va\in \Sigma^{*}M$. By Lemma \ref{lem: unique final state} and the definition of $\omega(ua)=\omega(va)$ we conclude that $\sigma(va)=(cz,z)$ with $s=\tr(cz)$ and $\sigma(ua)=(ch,h)$ with $s=\tr(ch)$. By definition of the action on $T(I)$ and equality (\ref{eq: final state}) we have
$$
(b,x,\omega(u))\cdot a=(c,s,\omega(\tau(ua)))=\left(c,s,\omega(ua)\setminus \{(c,s,f)\}\right)
$$
Therefore, since $\omega(va)\setminus \{(c,s,f)\}=\omega(ua)\setminus \{(c,s,f)\}$ and $va\in\Sigma^{*}M$ we may conclude by equality (\ref{eq: final state}) that 
$$
\omega(\tau(va))=\omega(va)\setminus \{(c,s,f)\}=\omega(ua)\setminus \{(c,s,f)\}=\omega(\tau(ua))
$$
holds. Hence, we get our claim 
$$
va\in I\left(c,s,\omega(ua)\right)=I((b,x,\omega(u))\cdot a).
$$
\item Suppose $ua\notin\Sigma^{*}M$. Since $\omega(u)=\omega(v)$ by Lemma \ref{lem: equality of omega} we get $\omega(ua)=\omega(u)\circ a=\omega(v)\circ a=\omega(va)$. In particular, we have $va\notin \Sigma^{*}M$. Since $v\in I(b,x,\omega(u))$ we have $\sigma(v)=(bz,y)$, $\sigma(va)=(bz,ya)$ and $x=\tr(by)$. 
Thus, by the definition of the action on $\mathcal{T}(I)$ we have:
$$
va\in I(b,\tr(bya),\omega(va))= I(b,x+a,\omega(va))=I(b,x+a,\omega(ua))=I\left((b,x,\omega(u))\cdot a\right)
$$
\end{itemize}
\end{proof}
We have the following theorem.
\begin{theorem}\label{theo: new reset decomposition}
Given a regular ideal $I$ on an alphabet with $|\Sigma|>1$, the finite family 
$$
\mathcal{F}=\left\{I(b,s,\omega(u)): \omega(u)\in \Omega(I), s\in\mathbb{Z}_{m}[\Sigma], b\in\Sigma\right\}
$$
is a reset left regular decomposition of $I$. 
\end{theorem}
\begin{proof}
Note that each element $u\in I$ has a uniquely determined tail structure $\sigma(u)=(bw,t)$, and so it gives rise to a unique triple $(b,\tr(bt),\omega(u))$. Hence, the family $\mathcal{F}$ forms a partition of $I$. Moreover, since $\sigma(vu)=\sigma(u)$ for all $v\in\Sigma^{*}$, $u\in I$, then every $I(b,s,\omega(u))$ is a left ideal, hence the above decomposition is formed by left ideals that by Lemma \ref{lem: action of left ideals} satisfies the condition that for any $a\in \Sigma$
$$
I(b,s,\omega(u))a\subseteq I\left((b,x,\omega(u))\cdot a\right)\in\mathcal{F}
$$
It remains to show the reset condition ii) of Definition \ref{defn: regular dec}, so let $v\in\Sigma^{*}$ such that
$$
Iv\subseteq I(b,s,\omega(z)), \mbox{ for some }I(b,s,\omega(z))\in\mathcal{F}
$$
We have to show that $v\in I$. This condition implies that for any $u\in I$ we have:
\begin{equation}\label{eq: condition reset}
\sigma(uv)=(bx,y), \mbox{ with }\tr(by)=s,\mbox{ and }\omega(uv)=\omega(z).
\end{equation}
Suppose, contrary to our claim, that $v\notin I$. Let $w\in I$ be an element of minimal length, i.e., $|w|=m=\|I\|$. Let $c\in\Sigma$ with $c\neq b$ (here it is important that $|\Sigma|>1$). Consider the series of words $wc^{k}\in I$. By the condition (\ref{eq: condition reset}) we have that $\sigma(wc^{k}v)=(bx_{k},y_{k})$ for all $k\ge 1$, and $by_{k}$ is not a suffix of $v$, for if $v\in I$, a contradiction. Thus, since $b\neq c$ we necessarily have that there are prefixes $w_{k}$ of $w$, $k\ge 1$, such that $w_{k}by_{k}=wc^{k}v$ (see the following figure).
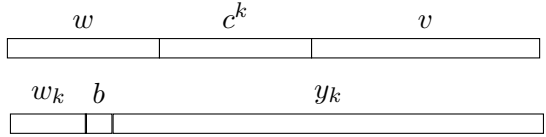
\begin{center}
\begin{tikzpicture}
\begin{scope}[xshift=0]
\node [draw,rectangle,minimum width=2cm,minimum height=0.2cm, label=$w$] {};
\end{scope}
\begin{scope}[xshift=2cm]
\node [draw,rectangle,minimum width=2cm,minimum height=0.2cm, label=$c^{k}$] {};
\end{scope}

\begin{scope}[xshift=4.5cm]
\node [draw,rectangle,minimum width=3cm,minimum height=0.2cm, label=$v$] {};
\end{scope}

\begin{scope}[yshift=-1cm, xshift=-0.46cm]
\node [draw,rectangle,minimum width=1.0cm,minimum height=0.2cm, label=$w_{k}$] {};
\end{scope}

\begin{scope}[yshift=-1cm, xshift=0.2cm]
\node [draw,rectangle,minimum width=0.35cm,minimum height=0.2cm, label=$b$] {};
\end{scope}

\begin{scope}[yshift=-1cm, xshift=3.22cm]
\node [draw,rectangle,minimum width=5.66cm,minimum height=0.2cm, label=$y_{k}$] {};
\end{scope}

\end{tikzpicture}
\end{center}
Therefore, by condition (\ref{eq: condition reset}) we have that $s=\tr(by_{k})$, for any $k\ge 1$. Let us denote by $z_{k}$ the suffixes of $w$ such that $w=w_{k}z_{k}$, for all $k\ge 1$, and by $n_{c}(u)$ the number of occurrences of the letter $c$ in the word $u$ module $m$. By counting the number of occurrences of the letter $c$, and from the condition $s=\tr(by_{k})$, $k\ge 1$, we easily deduce that for any pair of integers $\ell, \ell'$ we have:
\begin{equation}\label{eq: module condition}
s-n_{c}(v)=\ell+n_{c}(z_{\ell})=\ell'+n_{c}(z_{\ell'}) \mod\,m
\end{equation}
Let $R=\{n_{c}(z_{\ell})-n_{c}(z_{\ell'}) \mod\,m, \mbox{ for all }\ell, \ell'\ge 0\}$. Since $w$ contains the letter $b$, it is not difficult to see that $|R|<m$. Take any $r\in \mathbb{Z}_{m}\setminus R$. Since the parameter $\ell$ is arbitrary, then we may find two integers $\ell_{1}, \ell_{2}$ satisfying $\ell_{1}-\ell_{2}=r\mod\, m$. Hence, the two integers $\ell_{1}, \ell_{2}$ contradict equation (\ref{eq: module condition}). Therefore, $v\in I$ and this concludes the proof of the theorem. 
\end{proof}
From the previous theorem we have the following corollary. 
\begin{cor}\label{cor: another reset regular decomposition}
Let $I$ be a regular ideal on a non-unary alphabet $\Sigma$ with $k=|\Sigma|$, $m=\|I\|$, and let $n$ be the state complexity of the language of the set of the minimal reset words $M=\mathcal{M}(I)$. Then, there is a strongly connected synchronizing automaton $\mathcal{T}(I)=\la Q,\Sigma, \delta\ra$ with at most $(km^{k})2^{km^{k}n}$ states such that $\Syn(\mathcal{T}(I))=I$. Moreover, the construction of $\mathcal{T}(I)$ is effective and if the ideal $I$ is presented by its set of generators $M$ via the minimal DFA $\mathrsfs{B}$ recognizing $M$, then Algorithm \ref{alg:nlognAlg} returns $\mathcal{T}(I)$ in time $\mathcal{O}((k^{2}m^{k})2^{km^{k}n})$.
\end{cor}
\begin{proof}
The existence of the tail semiautomaton $\mathcal{T}(I)$ follows from Theorem \ref{theo: new reset decomposition} and \cite[Theorem 2.2]{ReRo16}. The bound follows from the cardinality of the set $T(I)=\Sigma\times \mathbb{Z}_{m}[\Sigma]\times\Omega(I)$ on which the tail action acts. The effectiveness of the construction of $\mathcal{T}(I)$ follows again from Theorem \ref{theo: new reset decomposition}, \cite[Theorem 2.2]{ReRo16} and Lemma \ref{lem: action of left ideals}. Note that $I(b,x,\omega(u))$ may be empty for some $(b,x,\omega(u))\in T(I)$, therefore to construct the automaton $\mathcal{T}(I)$, Algorithm \ref{alg:nlognAlg} finds all the states that are connected to $(a, \tr(aw), \omega(w))$, for some $aw\in M$ for which $I(a, \tr(aw), \omega(w))\neq \emptyset$. This is done in the first line of Algorithm \ref{alg:nlognAlg} where it is initialized $Q$ to the set $\{(a, \tr(aw), \omega(w))\}$. The rest of the algorithm is just a standard breadth first search in which each transition is recorded into the variable $\delta$.
\end{proof}
 \begin{algorithm}[t]
            \caption{The tail structure semiautomaton $\mathcal{T}(I)=\la Q,\Sigma, \delta\ra$}\label{alg:nlognAlg}
            \begin{algorithmic}
            \State{$Q\gets \{(a, \tr(aw), \omega(w))\}$, for some $aw\in M$, $a\in\Sigma$}
             \State{$A\gets Q$}
              \State{$B\gets \emptyset$}
              \While{$A\neq\emptyset$}
              \ForAll{$q\in A$}
               \ForAll{$a \in \Sigma$}
                  \If{$ q\cdot a\notin Q$}       \Comment{We have found a new state.}
                  \State{$Q\gets Q\cup\{q\}$}
                  \State{$B\gets B\cup\{q\}$} \Comment{Add to the states to check the new transitions.}
                \EndIf 
                \State $\delta\gets \delta \cup\{(q,a,q\cdot a)\}$\Comment{Note that the tail action $q\cdot a$ is computable}
                \EndFor
                \EndFor
                \State{$A\gets B$, $B\gets \emptyset$}
                \EndWhile
            \end{algorithmic}
          \end{algorithm}%
This last corollary partially answers to Problem 3 in \cite{ReRo16} since it presents a new construction of a reset left regular decomposition depending on the state complexity $n$ of the language $M$, and the obtained bound does not depend on a double exponential like in the bound 
\begin{equation}\label{eq: old bound}
m^{k2^{\ell}}\left(\sum_{t=2}^{\ell}m^{{\ell \choose t}}\right)^{2^{\ell}}
\end{equation}
presented in \cite[Corollary 3.5]{ReRo16}, where $m=\|I\|+1$, $k=|\Sigma|$ and $\ell$ is the state complexity of the reverse ideal $I^{R}=\Sigma^{*}M^{R}\Sigma^{*}$. Moreover, with the approach presented here we are able to explicitly construct a strongly connected automaton $\mathcal{T}(I)$ having $I$ as the set of reset words, while with the approach presented in \cite{ReRo13, ReRo16} the only way to explicitly build such an automaton would be to exhaustively check among all the automata with a number of states less than the bound stated in (\ref{eq: old bound}). 
\\
Comparing the bound in Corollary \ref{cor: another reset regular decomposition} and the bound (\ref{eq: old bound}) seems a non-easy task, mainly because the parameter $\ell$ is the state complexity of the reverse ideal $I^{R}$ while the parameter $n$ in Corollary \ref{cor: another reset regular decomposition} is the state complexity of $M$. By a result of Brzozowski et al. \cite[Theorem 6]{BrzJiLi} the state complexity $\sta(M)$ of $M$ differs from the state complexity $\sta(I)$ of the generated ideal $I=\Sigma^{*}M\Sigma^{*}$ by a polynomial function, namely $\sta(M)\le 3+(\sta(I)-1)(\sta(I)-2)/2$, and the bound is actually tight. When passing to the reversal, however, by \cite[Theorem 3]{BrzJiLi} we have $\sta(I)\le 2^{\sta(I^{R})-2}+1$ and this is also tight. From these two last bounds we get $\sta(M)\le 3+2^{\sta(I^{R})-2}(2^{\sta(I^{R})-2}-1)/2$, but we do not know whether or not this is tight (in \cite{BrzJiLi} the witnesses of the tightness of the two operations are different).

\end{document}